%% file: ms.tex
\documentclass{lmcs}
\pdfoutput=1

\usepackage{silence}
\usepackage{lastpage}
\lmcsdoi{16}{1}{14}
\lmcsheading{}{\pageref{LastPage}}{}{}%
{Jul.~24,~2019}{Feb.~13,~2020}{}

\usepackage{amsthm}
\usepackage{amsmath,amssymb}

\theoremstyle{plain}

\usepackage{xspace}

\usepackage{tabularx}
\usepackage{environ}
\usepackage{tikz}
\usetikzlibrary{arrows,automata,calc,positioning}

\newcommand{\MinVPA}{\textsc{MinVPA}\xspace}
\newcommand{\MinImmersion}{\textsc{MinImmersion}\xspace}

\makeatletter
\newcommand{\problemtitle}[1]{\gdef\@problemtitle{#1}}
\newcommand{\probleminput}[1]{\gdef\@probleminput{#1}}
\newcommand{\problemquestion}[1]{\gdef\@problemquestion{#1}}
\NewEnviron{problem}{
  \problemtitle{}\probleminput{}\problemquestion{}
  \BODY
  \par\addvspace{.5\baselineskip}
  \noindent
  \begin{tabularx}{\textwidth}{@{\hspace{\parindent}} l X c}
    \multicolumn{2}{@{\hspace{\parindent}}l}{\@problemtitle} \\
    \textbf{Input:} & \@probleminput \\
    \textbf{Question:} & \@problemquestion
  \end{tabularx}
  \par\addvspace{.5\baselineskip}
}
\makeatother

\newcommand{\Aa}{\mathcal{A}}
\newcommand{\Bb}{\mathcal{B}}
\newcommand{\Cc}{\mathcal{C}}
\newcommand{\tup}[1]{\langle#1\rangle}
\newcommand{\set}[1]{\{#1\}}
\renewcommand{\a}{\alpha}
\newcommand{\D}{\Delta}
\newcommand{\act}[1]{\stackrel{#1}{\longrightarrow}}

\def\ptime{\textsc{PTime}}
\def\np{\textsc{NP}}
\def\pspace{\textsc{PSPACE}}
\def\exptime{\textsc{ExpTime}}

\def\vpa{\textsc{VPA}}

\usepackage{subcaption}

\begin{document}

\title{Minimization of visibly pushdown automata is NP-complete}

\author[O.~Gauwin]{Olivier Gauwin\rsuper{a}}
\author[A.~Muscholl]{Anca Muscholl\rsuper{a}}
\author[M.~Raskin]{Michael Raskin\rsuper{b}}

\address{\lsuper{a}LaBRI, Univ.~of Bordeaux, France}
\address{\lsuper{b}TU Munich, Germany}
\thanks{Work done while the third author was affiliated
  with the University of Bordeaux.
  This work was partially supported by the ANR project DeLTA (ANR-16-CE40-0007). 
}

\keywords{visibly pushdown automata, minimization}

\begin{abstract}
We show that the minimization of visibly pushdown automata
is NP-complete.
This result is obtained by
introducing  immersions, that recognize multiple languages
(over a usual, non-visible alphabet)
using a common deterministic transition graph,
such that each language is associated with an initial state and a set
of final states.
We show that minimizing immersions is NP-complete,
and reduce this problem to the minimization of visibly pushdown automata.
\end{abstract}

\maketitle

\section{Introduction}

Visibly pushdown automata (\vpa) are a natural model for the control
flow of recursive programs and have tight connections with tree
automata and XML schemas. They were considered for parsing
algorithms~\cite{Mehlhorn80} under the name ``input-driven pushdown
automata'', and  shown to have better space complexity than
unrestricted pushdown automata. The name ``visibly pushdown automata''
is due to Alur and Madhusudan~\cite{AlurMadhusudan04}, who initiated
their study from the perspective of program verification, and
developed the theory in several directions
(see~\url{http://madhu.cs.illinois.edu/vpa/} for an exhaustive list of
results). In particular, they
showed that the class of visibly pushdown languages
shares many desirable properties with the class of regular languages,
like determinization, closure under boolean operations and the
existence of a Myhill-Nerode congruence that defines canonical
\vpa~\cite{AlurKumarMadhusudanV05}. However, the existence of
a canonical \vpa{} does not help
for minimization, in contrast to regular languages. Similarly
to many other more complex automata models, like automata over infinite words,
two-way automata, etc,
\vpa{} do not have unique minimal automata. Even worse, the canonical \vpa{} can be
exponentially larger than a minimal \vpa.  Therefore, the minimization
problem for \vpa, besides being very relevant in practice,  is also   very challenging.

\medskip

\emph{Minimization up to partitioning.}
Various minimization procedures have been proposed
for some subclasses of deterministic \vpa.
Most of them use a partitioning of the state space into \emph{modules}:
when the \vpa{} control is in a given module,
and a call occurs, then the matching return brings the control
back to the same module.
The first model implementing this idea are~\emph{single-entry \vpa}
(\textsc{Sevpa})~\cite{AlurKumarMadhusudanV05}, where
each module has  its own set of call symbols, and these sets are
disjoint.
Moreover, each module has a specific entry state:
whenever a call of module $m$ occurs,
 the \vpa{} switches to the entry  of  $m$.
For any fixed partition of call symbols~\cite{AlurKumarMadhusudanV05}
shows that
there is a unique minimal deterministic \textsc{Sevpa}, and
that it can be computed in polynomial time.
\emph{Multiple-entry \vpa} (\textsc{Mevpa})~\cite{KumarMadhusudanViswanathan06}
allow several possible states when entering the module,
but the symbol pushed on the stack by a call
depends only on the state, not on the call symbol.
\textsc{Mevpa} enjoy the same properties as
\textsc{Sevpa} in terms of minimization:
the minimal \textsc{Mevpa} is unique and computable in polynomial time.
The two models \textsc{Sevpa} and \textsc{Mevpa}
are subsumed by \emph{call-driven automata}
(\textsc{CDA})~\cite{ChervetWalukiewicz07}, for which states are
partitioned into modules, and a call
leads to a state that depends only  on the call symbol.
A restricted version of \textsc{CDA},
called \emph{expanded CDA} (\textsc{eCDA})~\cite{ChervetWalukiewicz07},
further requires that only one call symbol can enter each module.
Minimization of \textsc{eCDA} is easy,
it resembles the Myhill-Nerode construction.
A minimization procedure for \textsc{CDA} is obtained
by adapting that of \textsc{eCDA},
and generalizes the ones for \textsc{Sevpa} and \textsc{Mevpa},
in the sense that these ones can be retrieved
from the minimization of \textsc{CDA}.

The drawback of all the subclasses mentioned above (\textsc{Sevpa},
\textsc{Mevpa}, CDA and \textsc{eCDA}) is that there exist families of
languages for which the minimal \vpa{} within the respective class is
exponentially larger than some minimal \vpa.  \emph{Block \vpa}
(\textsc{BVPA})~\cite{ChervetWalukiewicz07} were proposed to overcome
this problem: for every \vpa, there exists an equivalent \textsc{BVPA}
of quadratic size, so \vpa{} can be minimized approximately via
\textsc{BVPA} minimization.  \textsc{BVPA} differ from \textsc{Sevpa}
in that the entry state is determined by the call symbol, but may also
depend on the current state.  There is a unique minimal \textsc{BVPA}
for a given visibly pushdown language, and this \textsc{BVPA} can be
computed in cubic time, up to some partition of the language.

So all the approaches for \vpa{} minimization rely
on a fixed partition, either of the state space, or of the language,
and the difficulty of minimization relies on finding
a good partition. Given a \textsc{BVPA} and two integers $k$ and $s$,
knowing if there is an equivalent \textsc{BVPA}
with $k$ modules, each of size at most $s$, is \np-complete~\cite{Duvignau12}.
\medskip

The main result of this paper is that \vpa{} minimization is
inherently difficult: we show that the problem is \np-complete. We obtain our result by
showing \np-hardness for the following problem about deterministic
finite state automata (DFA), that can be of independent interest: given
$n$ regular languages and a bound $N$, we ask if there exists some deterministic
transition graph $\Aa$ of size $N$ such that for every given language
we find a DFA accepting it by choosing
an initial state and a set of final states of $\Aa$. We refer to this
problem as \emph{immersion minimization}.

\medskip

\emph{Further related work.}
As for regular languages, finding a minimal \emph{non-deterministic} automaton
is computationally hard, namely
\exptime-complete for non-deterministic \vpa,~\cite{Duvignau12}
(hardness follows from
the universality of non-deterministic \vpa~\cite{AlurMadhusudan04}).
The paper~\cite{HeizmannSchillingTischner17} proposes an algorithm  for computing
\emph{locally minimal} non-deterministic \vpa, relying on a reduction
to Partial Max-SAT\@.
Results on the state complexity of \vpa{} with respect to determinization,
and various language operations are reported in the survey~\cite{OkhotinSalomaa14}.

Some problems similar to the minimization of immersions
also appear in the literature.
However, to our best knowledge, no straightforward reduction  exists
from one of these problems to the minimization of immersions. The
first problem is the minimization of non-deterministic finite
automata with limited non-determinism. Whereas the minimization of
arbitrary non-deterministic finite
automata is \pspace-complete, it becomes \np-complete for automata
that have a fixed number of initial states, and are otherwise
deterministic~\cite{Malcher04}. Further \np-completeness results for
minimization of automata with small degree of ambiguity are provided
in~\cite{BM12}. A seemingly close problem from computational biology is the \emph{shortest common superstring} problem,
which asks for the shortest string containing each string from a given
set as factor. This problem is known to be \np-complete~\cite{GareyJohnson90}.

Another related problem is the minimization of tree automata.
Indeed, a word over a visibly pushdown alphabet can be viewed as
the linearization of a tree, processed in a depth-first left-to-right traversal.
This corresponds to an unranked tree, i.e,  a finite ordered tree where the arity
of each node is arbitrary.
Several automata models exist for unranked trees, and the complexity of
minimization ranges between \ptime{} and
\np~\cite{MartensNiehren07}. However,
for each of these models, determinism does not correspond exactly
to that of \vpa, and minimization results do not transfer.

\section{Automata}

\subsection{Visibly pushdown automata}

A \emph{visibly pushdown alphabet}
$\widehat\Sigma=\Sigma_c\uplus\Sigma_r\uplus\Sigma_\ell$ is
a finite set of \emph{symbols} partitioned into
\emph{call symbols} in $\Sigma_c$,
\emph{return symbols} in $\Sigma_r$, and
\emph{internal symbols} in $\Sigma_\ell$.

A \emph{visibly pushdown automaton} (\vpa{} for short)
is a tuple $\Cc=\tup{\widehat\Sigma,Q, I,F,\Gamma,\Delta}$
where $\widehat\Sigma$ is a visibly pushdown alphabet,
$Q$ is a finite set of \emph{states},
$I\subseteq Q$ and  $F\subseteq Q$ are the sets of initial,
resp.~final states, and
$\Gamma$ is the (finite) stack alphabet.
The set $\D$  has three types
of transitions, depending on the type of the input symbol:
\emph{call} transitions
$\Delta_c\subseteq Q\times\Sigma_c\times Q\times\Gamma$ that push a
symbol on the stack,
\emph{return} transitions
$\Delta_r\subseteq Q\times\Sigma_r\times \Gamma\times Q$ that pop a
symbol from the stack, and
\emph{internal} transitions
$\Delta_\ell\subseteq Q\times\Sigma_\ell\times Q$ that leave
the stack unchanged.

A configuration of $\Cc$ is a pair $(q,\sigma)$
where $q\in Q$ is the current state and $\sigma\in\Gamma^*$
is the current stack content (the top of the stack is the rightmost symbol).
A transition $(q,\sigma)\act{a}_\Cc (q',\sigma')$ corresponds to one
of the following cases:
\begin{itemize}
\item $a\in\Sigma_c$ and $\sigma'=\sigma A$ for some $q,q',A$ with
$(q,a,q',A)\in\Delta_c$,
\item $a\in\Sigma_r$ and $\sigma=\sigma' A$ for some  $q,q',A$ with
$(q,a,A,q')\in\Delta_r$,
\item $a\in\Sigma_\ell$ and $\sigma=\sigma'$ for some $q,q'$ with
$(q,a,q')\in\Delta_\ell$.
\end{itemize}
Note that only return transitions can read the top stack symbol. The
transition relation of $\Cc$ extends to words from $\Sigma^*$ as
expected.
The \emph{language accepted by $\Cc$} is the set of words  $u$ such that $(q_0,\epsilon)\act{u}_\Cc (q_f,\sigma)$
with $q_0\in I$, $q_f\in F$ and $\sigma\in\Gamma^*$.
In particular, acceptance does not require
that the final configuration has an empty stack.
A \vpa{} is \emph{deterministic} if it has a single initial state,
$\Delta_c$ does not contain two rules
$(q,a,q_1,\gamma_1)$ and $(q,a,q_2,\gamma_2)$
with $(q_1,\gamma_1)\not=(q_2,\gamma_2)$,
$\Delta_r$ does not contain two rules
$(q,a,\gamma,q_1)$ and $(q,a,\gamma,q_2)$ with $q_1\not=q_2$, and
$\Delta_\ell$ do not contain two rules
$(q,a,q_1)$ and $(q,a,q_2)$ with $q_1\not=q_2$.

\medskip

\emph{Minimization.}
We measure the size of a \vpa{} by its number of states.
This will be  the parameter that we minimize.
Another choice could be
the size of the stack alphabet.
The stack alphabet can be actually bounded by $|Q||\Sigma_c|$,
as one can always choose it as $ Q \times \Sigma_c$,~\cite{ChervetWalukiewicz07}. The problem we consider here is the
following:

\begin{problem}
  \problemtitle{\MinVPA}
  \probleminput{deterministic \vpa{} $\Cc$ and integer $N$}
  \problemquestion{does a  deterministic \vpa{} $\Cc'$ of size $N$
    exist that accepts the same language
as $\Cc$?}
\end{problem}

The main result of the paper is:

\begin{thm}\label{th:min}
\MinVPA is NP-complete.
\end{thm}

\begin{proof}
Since equivalence of deterministic \vpa{} can be checked in polynomial
time~\cite{AlurMadhusudan04}, it is clear that \MinVPA belongs to \np.
We show hardness through an intermediate problem called \MinImmersion.
We first prove that \MinImmersion reduces to \MinVPA
(Proposition~\ref{prop:MinImmersion-to-MinVPA}),
and then show that \MinImmersion is \np-hard, by
reduction from 3-colorability of graphs (Section~\ref{sec:min-immersion}).
\end{proof}

\subsection{Immersions and VPA minimization}

An immersion is a variant of a deterministic finite state automaton
used to accept multiple regular languages.  We show in
this section that minimization of immersions reduces to minimization
of \vpa.

\emph{Sub-DFAs.}
Let $\Aa=\tup{Q,\Sigma,\act{}}$ be a finite, deterministic
transition graph, so $Q$ is a finite set of states, $\Sigma$ is the alphabet, and
$\mathop{\act{}}$ is a partial function from $ Q \times \Sigma$ to
$Q$. A sub-automaton (\emph{sub-DFA} for short) $\Bb$ of $\Aa$ is a
tuple $\tup{Q,\Sigma,\act{},q_0,F}$. So $\Bb$ is a deterministic
finite automaton (DFA for short) obtained by equipping $\Aa$ with an
initial state $q_0$ and a set $F \subseteq Q$ of final states. Note
that all DFA in our constructions will have partially-defined transition
functions. However, the results are not affected by introducing a dead-state. 

\emph{Immersions.} Given $n$ languages $L_1,\ldots,L_n\subseteq\Sigma^*$,
an \emph{immersion} for $L_1,\ldots,L_n$
consists of a finite \emph{deterministic} transition graph $\Aa$,
and $n$ sub-DFA $\Bb_1,\ldots,\Bb_n$ of $\Aa$
such that $L(\Bb_i)=L_i$, for all $1\le i\le n$.
The size of the immersion is the number of states of
its transition graph $\Aa$.
For convenience we usually just write $\Aa$ for the immersion,
omitting the initial/final states.

\input{fig-immersion}

Consider for instance the two languages $L_1=a^+$ and $L_2=a^*b$.
A possible immersion for these two languages is obtained by taking
the disjoint union of two DFA, one for each language.
This is illustrated in Figure~\ref{fig:immersion-size-4},
and yields an immersion with 4 states.
A smaller immersion is obtained by merging the states $2,3$ in
Figure~\ref{fig:immersion-size-4},
as depicted in Figure~\ref{fig:immersion-minimal}.
The resulting immersion has 3 states, and is minimal for
$L_1,L_2$. 
Another immersion with three states is obtained by
merging the states $1,4$. The example shows that, in general, minimal
immersions are not unique, as it is already the case
for \vpa.

\begin{problem}
  \problemtitle{\MinImmersion}
  \probleminput{DFA $\Aa_1,\ldots,\Aa_n$, and integer $N$}
  \problemquestion{is there some immersion of size at most $N$ for
    $L(\Aa_1),\dots,L(\Aa_n)$?}
\end{problem}

\begin{prop}%
\label{prop:MinImmersion-to-MinVPA}
\MinImmersion reduces in polynomial time to \MinVPA\@.
\end{prop}

\begin{proof}
Let $\Aa_1,\dots,\Aa_n$ be DFA over the alphabet $\Sigma_\ell$, and let
$L_i=L(\Aa_i)$ for every $i$.
We show that there exists an immersion of size $k$ for $L_1,\ldots,L_n$
if and only if there exists a deterministic \vpa{} of size $k+2$
for the language $K = \bigcup_{i=1}^n c_i L_i r$,
where $\Sigma_c=\set{c_1,\dots,c_n}$ and  $\Sigma_r=\set{r}$.

Consider an immersion of size $k$ for $L_1,\ldots,L_n$
with finite, deterministic transition graph $\Aa=\tup{Q,\Sigma,\act{}}$,
and sub-DFA $\Bb_1,\ldots,\Bb_n$ of $\Aa$
such that $L(\Bb_i)=L_i$, for all $1\le i\le n$.
Let $q_i$ and $F_i$ denote the initial state and the final states of $\Bb_i$,
respectively.
From $\Aa$  we immediately get a deterministic \vpa{} $\Cc$ for $K$ by  letting
$\Cc = (\widehat\Sigma,Q\uplus\{q_0,q_f\},\{q_0\},\{q_f\},\Gamma,\Delta)$,
with stack alphabet $\Gamma=\set{1,\dots,n}$, and $\Delta$ as follows:

\begin{itemize}
\item $\Delta_c = \{ (q_0, c_i, q_i, i) \mid 1 \le i \le n \}$,
\item $\Delta_\ell = \mathop{\act{}}$,
\item $\Delta_r = \{ (q, r, i, q_f) \mid q\in F_i,1 \le i \le n  \}$.
\end{itemize}

\noindent
Conversely, assume there is some deterministic \vpa{}
$\Cc=(\widehat\Sigma,Q,\{q_0\},F,\Gamma,\Delta)$ of size $k+2$ for
$K=\bigcup_{i=1}^n c_i L_i r$.
This language is included in $\Sigma_c\Sigma_\ell^*\Sigma_r$,
so we can assume that $\Cc$ has a single
final state, that we call $q_f$, and which has no outgoing transitions.
Let $q_i$ denote the (unique) state of $\Cc$ such that
$(q_0,c_i,q_i,A_i)\in\Delta_c$, for some $A_i\in\Gamma$.
We can also assume that $i$ is used instead of $A_i$ in
these rules, as they are the only rules in $\Delta_c$.
We define an immersion for the languages $L_1,\dots,L_n$
as the transition graph $\Aa=\tup{Q_\Aa,\Sigma,\act{}}$,
where $Q_\Aa = Q\setminus\{q_0,q_f\}$ and $\mathop{\act{}} =
\Delta_\ell$.
The sub-DFA $\Bb_1,\ldots,\Bb_n$ associated with this immersion
are obtained by setting the initial state of $\Bb_i$ to $q_i$, and setting $q \in F_i$ if $(q, r, i, q_f)\in \Delta_r$.
It is clear that $\Bb_i$ accepts precisely the words $w \in\Sigma_\ell^*$ such that
$c_i w r \in L(\Cc)$. So $L(\Bb_i) = L_i$
for every $1\le i\le n$.
\end{proof}

\section{\MinImmersion is \np-complete}%
\label{sec:min-immersion}

It is clear that \MinImmersion is in \np. We show \np-hardness by
a reduction from 3-colorability. Let $G=(V,E)$ be an undirected graph
with vertex set $V=\set{1,\dots,n}$ and edge set
$E \subseteq V^2 \setminus \set{(i,i) \mid i \in V}$. We ask whether
there is a coloring $c:V \to \set{0,1,2}$ such that $c(i) \not= c(j)$,
for every $(i,j) \in E$.

Before we define the DFA $\Aa_1,\ldots,\Aa_n$ we need some notations. Let $m = 2n(n-1)+2$. We fix a set $P=\set{p_1,p_2,p_3,q_1,q_2}$
of five distinct prime numbers $p$ such that $3n<p \le c\cdot n$, for
some suitable constant\footnote{Recall that Chebyshev's theorem says
  that there is always at least one prime
  between $n$ and $2n$.} $c$, such that no $p \in P$ divides $m$.
Let also
$N=3m+p_1+p_2+p_3+q_1+q_2$. Note that $6n^2 < N < 9n^2$, for $n$
sufficiently large.

\emph{Notations.} The alphabet used in the following for the DFA
$\Aa_i$ is  $\Sigma=\set{0,1}$. A path in some transition
graph of the form $s_1 \act{0} s_2 \act{0} \cdots \act{0}
s_n$  will be called simply a
\emph{path}. Similarly, a \emph{cycle} is a path as above, with
$s_1=s_n$. For any path $s_1 \act{0} s_2 \act{0} \cdots \act{0}
s_n$ we say 
that $s_n$ is \emph{0-reachable} from $s_1$,
and $s_1$ is
\emph{co-0-reachable} from $s_n$.
A \emph{$k$-cycle} denotes a cycle of length
$k$. A \emph{1-transition} is a transition labeled by 1.

We fix in the following a bijection between the set $\set{2k+1 : 1 \le k \le
  n(n-1)}$ and the set of ordered pairs of vertices $\set{(i,j) : i,j \in
  V, i \not= j}$. Hereby we denote by $\tup{i,j}$ the integer encoding the
pair $(i,j)$ w.r.t.~this fixed bijection.

We are now ready to define the DFA $\Aa_i$, where $i \in V$ is a
vertex of the given graph. The language $L_i$ of the DFA $\Aa_i$ will
be a subset of
$0^* 1 0^*$. Informally, $\Aa_i$ consists of an $m$-cycle (called ``dispatch'' cycle), such that from some of
the vertices of this cycle there is a 1-transition to some
$p$-cycle (called ``counting'' cycle) with $p \in P$. Each $p$-cycle
has a designated ``entry'' node, and all 1-transitions into the cycle
point to this node.
Assuming that the vertices of the $m$-cycle are numbered successively
$1,\dots,m$, with $1$ being the initial state, the DFA $\Aa_i$ has the following transitions:
\begin{enumerate}
\item From vertex 1 there is a 1-transition to a
$p_1$-cycle, and from vertex 2 there is a 1-transition to a
$p_2$-cycle.
\item From every other \emph{even} vertex there is a 1-transition to a
  $p_3$-cycle.
  \item Each of the remaining $n(n-1)$ \emph{odd} vertices
    is of the form $\tup{j,k}$, according to the
    bijection fixed above. The transitions out of these vertices are
    the following:

    \begin{itemize}
    \item  Each odd vertex $\tup{i,j}$ has a
      1-transition to a $q_1$-cycle.
    \item Each odd vertex  $\tup{j,i}$ with
       $\set{i,j} \notin E$, has a
      1-transition to a $q_1$-cycle.
      \item Each odd vertex  $\tup{j,i}$
      with $\set{i,j} \in E$, has a
      1-transition to a $q_2$-cycle.
    \end{itemize}
\end{enumerate}
Note that there is no 1-transition outgoing from vertices $\tup{j,k}$
where $j\not=i$ and $k\not=i$.
 As already mentioned, the initial state of $\Aa_i$ is the vertex 1 of the
 $m$-cycle. The final states are all the target states of the
 1-transitions. Figure~\ref{fig:A} shows an example $\Aa_i$.

\input{fig-dispatch}

\begin{rem}\label{rem:unary}
Note that any DFA accepting ${(0^p)}^*$ must contain
      a cycle of length divisible by $p$, if $p >1$ is a prime.
\end{rem}

Let $p$ be a prime from $P$. A vertex $s$ of a
transition graph
$\Aa$ over $\Sigma=\set{0,1}$ is called a \emph{$p$-vertex} if there is a sub-DFA of $\Aa$ for
the language $1 {(0^p)}^*$ with initial state $s$.

  \begin{lem}\label{lem:cycles}
     Let $\Aa$ be a minimal immersion for $L_1,\dots,L_n$ of
     size at most $N$, and
  let $C$ be a $k$-cycle of $\Aa$. Then exactly one of the two
  following cases holds:
  \begin{enumerate}
  \item $C$ contains at least one $p_1$-vertex and  $k$ is divisible by $m$.
  \item $k$ is divisible by some unique prime $p \in P$.
    \end{enumerate}
     \end{lem}

  \begin{proof}
    By assumption there is some sub-DFA $\Bb_i$ of $\Aa$ accepting
    $L_i$, for every $i$. Note first that, by minimality of $\Aa$,
    every vertex of $\Aa$ is either 0-reachable from the initial state
    of some $\Bb_i$, or co-0-reachable from a final state of some
    $\Bb_i$. This will ensure that one of the two cases in the statement of the
    lemma holds for any cycle.

    Recall that immersions were defined as \emph{deterministic}
    transition graphs. 
    Using the assumption $|\Aa|\leqslant N$, note
    that a vertex of $\Aa$ cannot
    be both a $p$-vertex and a $p'$-vertex, for two different primes
    $p,p'$ from $P$. Then otherwise the size of $\Aa$ would be at
    least $9n^2>N$, which is a contradiction.

    Let us denote a vertex $s$ of $\Aa$ as \emph{special} if $s$ is
    $p_1$-vertex and  $s\act{0}
    s'$, with $s'$ being $p_2$-vertex.

    The first case in the statement corresponds to $C$ being 0-reachable
    from the initial state of some $\Bb_i$. Clearly, $C$ needs to have
    at least one special vertex $s$. Note also that any vertex $t$
    such that $s \act{0^{m\cdot j}} t$, where $j \ge 0$, must
    be special, because the length of the dispatch cycle is
    $m$. Assume by contradiction that $k$ is not divisible by $m$, and
    let $d = k \pmod m$. If $d$ is odd, then it follows from the
    previous remark that $C$ must contain a $p_1$-vertex that is at
    the same time a $p_2$-vertex or a $p_3$-vertex. If $d$ is even and not zero, then
    similarly, $C$ must contain a vertex that is at the same time a
    $p_2$-vertex and a $p_3$-vertex. So in both cases we obtain a contradiction
    to $|\Aa| \le N$, as already noted.

    The second case is where $C$ is co-0-reachable from a final state of
    some $\Bb_i$. Here, $k$ must be divisible by some
    $p \in P$. This prime is unique, as already observed.

    We argue finally that the two cases are mutually exclusive. If $k$
    were both divisible by $m$ and by $p \in P$ then $k>pm >N$, since
    $p$ does not divide $m$. But this is again a contradiction to $|\Aa| \le N$.
  \end{proof}

  Assume that $\Aa$ is a minimal immersion for $L_1,\dots,L_n$ of
     size at most $N$.  From Lemma~\ref{lem:cycles} we deduce that the vertex
    set of $\Aa$ is the disjoint union of two sets $V_1,V_2$, such that:
    \begin{itemize}
    \item Transitions within each $V_i$ ($i \in\set{1,2}$) are labeled
      only by 0.
      \item Transitions from $V_1$ to $V_2$ are labeled only by 1.
      \item There are no transitions from $V_2$ to $V_1$.
    \end{itemize}

   \noindent
   To see this, let us ignore the 1-labeled transitions of $\Aa$. Then
   we obtain a
    disjoint union of graphs (transitions are labeled only by
    0s). Each such graph consists of a cycle, plus some simple paths
    reaching the cycle.  From Lemma~\ref{lem:cycles} we know that each cycle  is used either
    to ``dispatch'' (case 1) or to ``count'' modulo some prime (case 2),
    and that the two cases are mutually exclusive. So by  the
    minimality of $\Aa$ we can conclude  that 1-labeled transitions of
    $\Aa$ go only from ``dispatch'' cycles to
    ``counting'' cycles.
    Again by minimality we can bound  the size and number
    of cycles in $V_1$ and $V_2$:

\begin{lem}\label{lem:disjoint-cycles}
  Assume that $\Aa$ is a minimal immersion for $L_1,\dots,L_n$ of
  size at most $N$. Then  the vertex
    set of $\Aa$ is the disjoint union of two sets $V_1,V_2$ as above, such that:
   \begin{enumerate}
  \item $V_1$ consists of at most three $m$-cycles.
    \item $V_2$ consists of $p$-cycles, one for each $p \in P$.
  \end{enumerate}
\end{lem}

\begin{proof}
  By Lemma~\ref{lem:cycles} we know that
  $V_2$ contains at least $|P|$  cycles, one for each $p \in P$. By
  minimality of $\Aa$, $V_2$ has exactly one $p$-cycle, for each $p
  \in P$.

  Now we consider $V_1$. By Lemma~\ref{lem:cycles}
  the cycles of $V_1$ have length divisible by $m$. As before, by
  minimality $V_1$ is a disjoint union of cycles. Each cycle $C$ has the
  property that it accepts some language $\set{u \in 0^* : u1v \in L_i \text{
      for some $v \in 0^*$}}$ from one of the $p_1$-vertices of
  $C$. In particular, $C$ is equal to ${C'}^j$ for some dispatch cycle
  $C'$ of one of the $L_i$ and some $j \ge 1$. By
  minimality of $\Aa$  we obtain that $C=C'$ and $j=1$. Finally, by
  the choice of $N$, we conclude that $V_1$ consists of at most three $m$-cycles.
\end{proof}

From Lemma~\ref{lem:disjoint-cycles} we see that each of the
sub-DFA for any of the $L_i$ consists of one of the $m$-cycles in $V_1$, with the
$p_1$-vertex as initial state, together with transitions labeled by 1
to the required $p$-cycles in $V_2$.

\begin{lem}\label{lem:red}
  The graph $G$ is 3-colorable if and only if there is some  minimal immersion
  for $L_1,\dots,L_n$ of size at most $N$.
\end{lem}

\begin{proof}
  Let us first assume that $G$ is 3-colorable. Then we argue that
  $\Aa$ can be built from at most three dispatch cycles $C_0,C_1,C_2$,
  one for each color 0,1 and 2 (together with
  $p$-cycles, one for each $p \in P$). Cycle $C_\a$ can be used for all
  vertices $i \not= j$ of color $\a$, since they are pairwise
  unconnected. To see this, note that vertex $\tup{i,j}$ of $C_\a$ is a
  $q_1$-vertex according to the definition of $\Aa_i$; and
  $\tup{i,j}$ is also $q_1$-vertex according to  $L_j$, since
  $\set{i,j} \notin E$.

  Conversely, if $\Aa$ has size at most $N$ then by
  Lemma~\ref{lem:disjoint-cycles} there are at most three $m$-cycles
  $C_0,C_1,C_2$ in $\Aa$. We color vertex $i$ by $\a$ if the sub-DFA for
  $L_i$ uses $C_\a$. This  coloring is proper, because if the sub-DFA
  for $L_i,L_j$  both use the same dispatch cycle, then $\set{i,j} \notin E$
  since otherwise vertex  $\tup{i,j}$ would be both a $q_1$- and a
  $q_2$-vertex, contradicting Lemma~\ref{lem:disjoint-cycles}.
\end{proof}

Lemma~\ref{lem:red} yields finally the claimed result, and also the proof of
Theorem~\ref{th:min}:

\begin{thm}
   \MinImmersion is  \np-complete.
\end{thm}

\section*{Conclusions}

We have shown that the \vpa{} minimization is intrinsically difficult,
by exhibiting an \np-lower bound.
A minor modification of the construction reduces
approximation of the chromatic number
to approximating the minimal size of an equivalent \vpa{},
thus
any constant-factor approximation of \vpa{} minimisation
is \np-hard.
Our result raises the quest for efficient
implementations of  SAT-based
minimization algorithms for \vpa.

\bibliographystyle{abbrv}
\bibliography{ms}

\end{document}

%% file: fig-immersion.tex
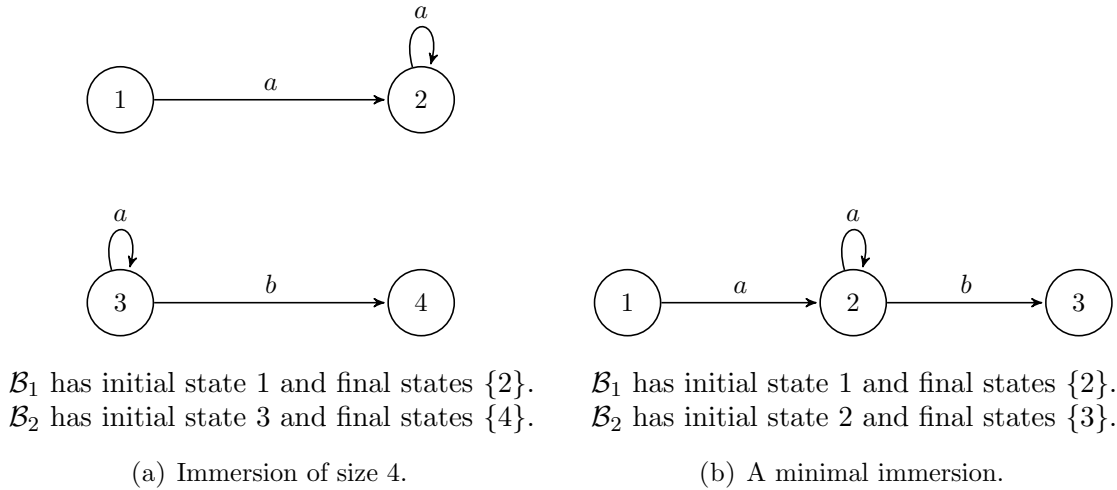
\begin{figure}
  \begin{subfigure}[b]{0.48\textwidth}
    \begin{center}
    \begin{tikzpicture}[->,>=stealth',shorten >=1pt,auto,node distance=4cm,
        semithick,scale=.5]
      \small

      \node[state] (q0) {1};
      \node[state] (q1) [right of=q0] {2};
      \path (q0) edge node {$a$} (q1)
      (q1) edge [loop above] node {$a$} (q1);

      \node[state] (q2) [below of=q0,yshift=1.3cm] {3};
      \node[state] (q3) [right of=q2] {4};
      \path (q2) edge [loop above] node {$a$} (q2)
      (q2) edge node {$b$} (q3);

    \end{tikzpicture}
    \\[1em]
    \begin{tabular}{c}
    $\Bb_1$ has initial state $1$ and final states $\{2\}$. \\
    $\Bb_2$ has initial state $3$ and final states $\{4\}$.
    \end{tabular}
    \end{center}
    \caption{Immersion of size 4.}%
    \label{fig:immersion-size-4}
  \end{subfigure}\hspace{3mm}
  \begin{subfigure}[b]{0.48\textwidth}
    \begin{center}
    \begin{tikzpicture}[->,>=stealth',shorten >=1pt,auto,node distance=3cm,
        semithick]
      \small

      \node[state] (q0) {1};
      \node[state] (q1) [right of=q0] {2};
      \node[state] (q2) [right of=q1] {3};
      \path (q0) edge node {$a$} (q1)
      (q1) edge [loop above] node {$a$} (q1)
      (q1) edge node {$b$} (q2);
    \end{tikzpicture}
    \\[1em]
    \begin{tabular}{c}
    $\Bb_1$ has initial state $1$ and final states $\{2\}$. \\
    $\Bb_2$ has initial state $2$ and final states $\{3\}$.
    \end{tabular}
    \end{center}
    \caption{A minimal immersion.}%
    \label{fig:immersion-minimal}
  \end{subfigure}
  \caption{Two immersions for the languages $L_1=a^+$ and $L_2=a^*b$.}%
  \label{fig:immersion}
\end{figure}


%% file: fig-dispatch.tex
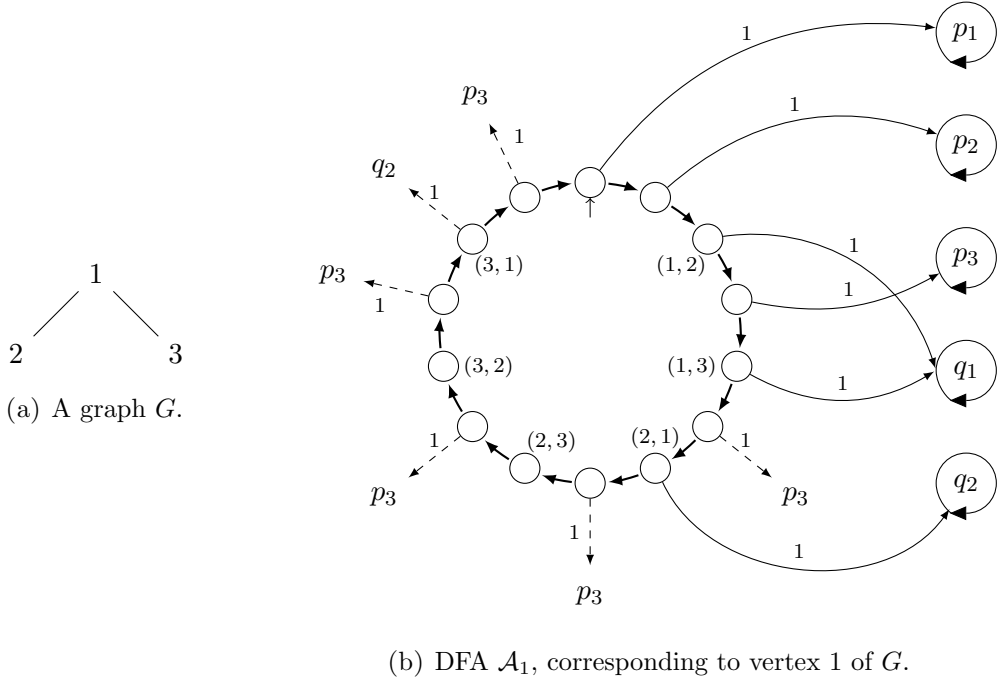
\begin{figure}[t]
  \centering%
  \begin{subfigure}{0.25\textwidth}
  \begin{center}
  \begin{tikzpicture}[node distance=1.5cm]
    \begin{scope}
      \node(v1) at (0,0) {$1$};
      \node(v2) [below left of=v1] {$2$};
      \node(v3) [below right of=v1] {$3$};
      \path (v1) edge (v2);
      \path (v1) edge (v3);
    \end{scope}
    \end{tikzpicture}
    \end{center}
    \caption{A graph $G$.}
  \end{subfigure}
  \begin{subfigure}{0.7\textwidth}
    \begin{center}
    \begin{tikzpicture}[node distance=1.5cm]
    \begin{scope}[xshift=6cm]
    \tikzstyle{dispatch state} = [draw,shape=circle,minimum size=.4cm]
    \tikzstyle{dispatch transition} = [<-, >=latex,thick]
    \tikzstyle{p label} = [draw=none,shape=circle,minimum size=.8cm]
    \tikzstyle{p cycle} = [->, >=triangle 45 reversed]
    \tikzstyle{1-transition} = [->, >=latex]
    \tikzstyle{ghost transition} = [1-transition,dashed]
    \tikzstyle{ghost label} = [1-transition,dashed]
    \tikzstyle{bijection} = [draw=none, node distance=.6cm]
    \tikzstyle{anch}=[draw=none, node distance=.1cm, inner sep=0mm]

    \def \n {14}
    \def \radius {2cm}
    \def \margin {7} 
    \foreach \s in {1,...,\n} {
      \def\p{\n-\s+1}
      \node[dispatch state] (disp\s) at ({360/\n * (\p) + 90}:\radius) {};
      \draw[dispatch transition] ({360/\n * (\p - 1)+\margin+90}:\radius)
      arc ({360/\n * (\p - 1)+90+\margin}:{360/\n * (\p)-\margin+90}:\radius);
    }

    \node[p label] (p1) at (5,4) {$p_1$};
    \draw[p cycle] (p1) +(0,-.4cm) arc (-90:270:.4cm);
    \node[p label] (p2) [below of=p1] {$p_2$};
    \draw[p cycle] (p2) +(0,-.4cm) arc (-90:270:.4cm);
    \node[p label] (p3) [below of=p2] {$p_3$};
    \draw[p cycle] (p3) +(0,-.4cm) arc (-90:270:.4cm);
    \node[p label] (q1) [below of=p3] {$q_1$};
    \draw[p cycle] (q1) +(0,-.4cm) arc (-90:270:.4cm);
    \node[p label] (q2) [below of=q1] {$q_2$};
    \draw[p cycle] (q2) +(0,-.4cm) arc (-90:270:.4cm);
    \def \gradius {3.5cm}
    \foreach \s in {1,...,\n} {
      \def\p{\n-\s+1}
      \node[p label] (gdisp\s) at ({360/\n * (\p) + 90}:\gradius) {};
    }

    \path[1-transition] (disp1) edge [bend left] node [above] {\scriptsize $1$} (p1);
    \path[1-transition] (disp2) edge [bend left] node [above] {\scriptsize $1$} (p2);
    \path[1-transition] (disp4) edge [bend right=20] node [above] {\scriptsize $1$} (p3);
    \node[ghost label] at (gdisp6) {$p_3$};
    \path[1-transition] (disp6) edge [ghost transition] node [above] {\scriptsize $1$} (gdisp6);
    \node[ghost label] at (gdisp8) {$p_3$};
    \path[1-transition] (disp8) edge [ghost transition] node [left] {\scriptsize $1$} (gdisp8);
    \node[ghost label] at (gdisp10) {$p_3$};
    \path[1-transition] (disp10) edge [ghost transition] node [above] {\scriptsize $1$} (gdisp10);
    \node[ghost label] at (gdisp12) {$p_3$};
    \path[1-transition] (disp12) edge [ghost transition] node [below left] {\scriptsize $1$} (gdisp12);
    \node[ghost label] at (gdisp14) {$p_3$};
    \path[1-transition] (disp14) edge [ghost transition] node [above right] {\scriptsize $1$} (gdisp14);
    \node[anch] (ghostq1) at  (4.6,-.5) {};
    \path[1-transition] (disp3) edge [bend left=40] node [above] {\scriptsize $1$} (ghostq1);
    \path[1-transition] (disp5) edge [bend right] node [above] {\scriptsize $1$} (ghostq1);
    \path[1-transition] (disp7) edge [bend right=60] node [above] {\scriptsize $1$} (q2);
    \node[ghost label] at (gdisp13) {$q_2$};
    \path[1-transition] (disp13) edge [ghost transition] node [above] {\scriptsize $1$} (gdisp13);

    \node[bijection, node distance=.5cm] (b3) [below left of=disp3] {\scriptsize $(1,2)$};
    \node[bijection] (b5) [left of=disp5] {\scriptsize $(1,3)$};
    \node[bijection, node distance=.4cm] (b7) [above of=disp7] {\scriptsize $(2,1)$};
    \node[bijection, node distance=.5cm] (b9) [above right of=disp9] {\scriptsize $(2,3)$};
    \node[bijection] (b11) [right of=disp11] {\scriptsize $(3,2)$};
    \node[bijection, node distance=.5cm] (b13) [below right of=disp13]
    {\scriptsize $(3,1)$};

    \node[bijection] (fake) [below of =disp1] {};
    \path[->] (fake) edge (disp1);
    \end{scope}
  \end{tikzpicture}
  \end{center}
  \caption{DFA $\Aa_1$, corresponding to vertex $1$ of $G$.}
  \end{subfigure}

  \caption{The dispatch cycle (thick) and the counting
    cycles of $\Aa_1$ ($p_1,p_2,p_3,q_1,q_2$). Dashed edges point to
    one of counting cycles. Unlabeled edges are 0-transitions. The
    final states are all target states of the 1-transitions.}%
  \label{fig:A}
\end{figure}
